\documentclass[conference]{IEEEtran}

\usepackage{cite}

\ifCLASSINFOpdf
   \usepackage[pdftex]{graphicx}
   \DeclareGraphicsExtensions{.pdf,.jpeg,.png}
\else
\fi

\usepackage{algorithm}
\usepackage{algorithmic}

\usepackage{subfigure}
\usepackage{tabu} 
\usepackage{multirow}
\usepackage{textcase}
\usepackage{tikz}
\usepackage{pgfplots}
\usepackage{subfigure}
\usepackage{amsmath}
\usepackage{amssymb}
\usepackage{bold-extra}
\usepackage{color}
\usepackage{graphicx}
\usepackage{amsthm}
\usepackage{csquotes}
\usepackage{url}
\usepackage[T1]{fontenc}
\usepackage[utf8]{inputenc}
\usetikzlibrary{calc}
\usetikzlibrary{positioning}
\usepackage[absolute,overlay]{textpos} 
\usepackage{latexsym}
\usepackage{csquotes}
\usepackage{color}
\usepackage{makecell}
\usepackage{enumitem}
\usepackage{diagbox}

\newcommand{\revise}[1]{\textcolor[rgb]{0.00,0.00,0.00}{#1}}

\tikzset{elegant/.style={smooth,thick,samples=50,cyan}}
\tikzset{eaxis/.style={->,>=stealth}}

\newcommand{\critical}{L}
\newcommand{\sumC}{C}
\newcommand{\taskset}{\tau}

\newtheorem{definition}{Definition}
\newtheorem{lemma}{Lemma}

\newtheorem{theorem}{Theorem}

\newcommand{\interfering}{\mathcal{I}}
\newcommand{\jobset}{\mathcal{J}}
\newcommand{\work}{W}

\begin{document}

\IEEEoverridecommandlockouts

\title{Schedulability Bounds for Parallel Real-Time Tasks under Global Rate-Monotonic Scheduling}

\author{\IEEEauthorblockN{Xu Jiang$^{1}$, Nan Guan$^2$, Maolin Yang$^3$, Yue Tang$^1$, Wang Yi$^{1,4}$}
	\IEEEauthorblockA{
		~~\\
		$^1$ Northeastern University, China\\
		$^2$ The Hong Kong Polytechnic University, Hong Kong \\
		$^3$ University of Electronic Science and Technology of China, China\\
		$^4$ Uppsala University, Sweden		    
	   }}

\maketitle

\begin{abstract}

Schedulability bounds not only serve as efficient tests to decide schedulability of real-time task systems, but also
reveal insights about the worst-case performance of scheduling algorithms. Different from sequential real-time task systems for which utilization is a suitable metric to develop schedulability bounds, schedulability of parallel real-time tasks depends on not only utilization but also the workload graph structure of tasks, which can be well represented by the tensity metric. In this paper, we develop new analysis techniques for parallel real-time tasks systems under Global Rate-Monotonic (G-RM) scheduling and obtain new results on schedulability bounds based on these two metrics: utilization and tensity. First, we develop the first utilization-tensity bound for G-RM. Second, we improve the capacity augmentation bound of G-RM from the best known value 3.73 to 3.18. These schedulability bounds not only
provide theoretical insights about real-time performance of G-RM, but also serve as highly efficient schedulability tests, which are particularly suitable to design scenarios in which detailed task graph structures are unknown or may change at run-time. Experiments with randomly generated task sets show that our new results consistently outperform the state-of-the-art with a significant margin under different parameter settings.
\end{abstract}

\IEEEpeerreviewmaketitle
\section{Introduction}

Schedulability bound is a well-established concept in real-time scheduling theory, which can be used as not only a simple and practical way to test the schedulability of real-time task sets, but also a good quantitative metric to indicate the worst-case performance of different scheduling strategies and provide insights about their performance bottlenecks. The most well-known schedulability bound would be Liu and Layland’s \emph{utilization bound} for \emph{Rate-Monotonic} (RM) scheduling algorithm on single-processors developed in 1970’s \cite{C1973Scheduling}. Since then, various of schedulability bounds have been developed for different scheduling algorithms with different task and processing platform models. 

Today, multi-core processors are more and more widely used in real-time systems, to meet the rapidly increasing requirements in high performance and low power consumption. Software must be parallelized to fully utilize the computation power of multi-cores. Therefore, it requires to upgrade the classical real-time scheduling theory from sequential tasks to the parallel task setting. The Directed Acyclic Graph (DAG) task model is a general representation of parallel tasks. There have been increasing research interests on real-time scheduling and analysis of DAG tasks recently \cite{li2013outstanding, bonifaci2013feasibility,melani2017schedulability,baruah2014improved,baruah2012generalized,chwa2013global,jiang2016decomposition,jiang2017,sun2017scheduling}, but the field is still far from mature.

While utilization has been proved to be a suitable metric to examine the schedulability of sequential tasks, this is not the case for parallel DAG task systems. A DAG task system may be unschedulable by any scheduling approach with arbitrary low utilization, even when the system only consists of one DAG task. This is because the difficulty of scheduling a DAG task depends on not only the total workload, but also the structural constraint of the tasks’ workload. Previous work has shown that, in additional to utilization, another metric \emph{tensity}, the ratio between the longest path length in the graph and the deadline, also captures important feature of the DAG structure of tasks and plays an important role in their schedulability. Based on this observation, two types of schedulability bounds have been proposed and studied for DAG task systems: the \emph{capacity augmentation bound} and the \emph{utilization-tensity bound}.

\begin{figure} 
	\centering
	\includegraphics[height =5cm]{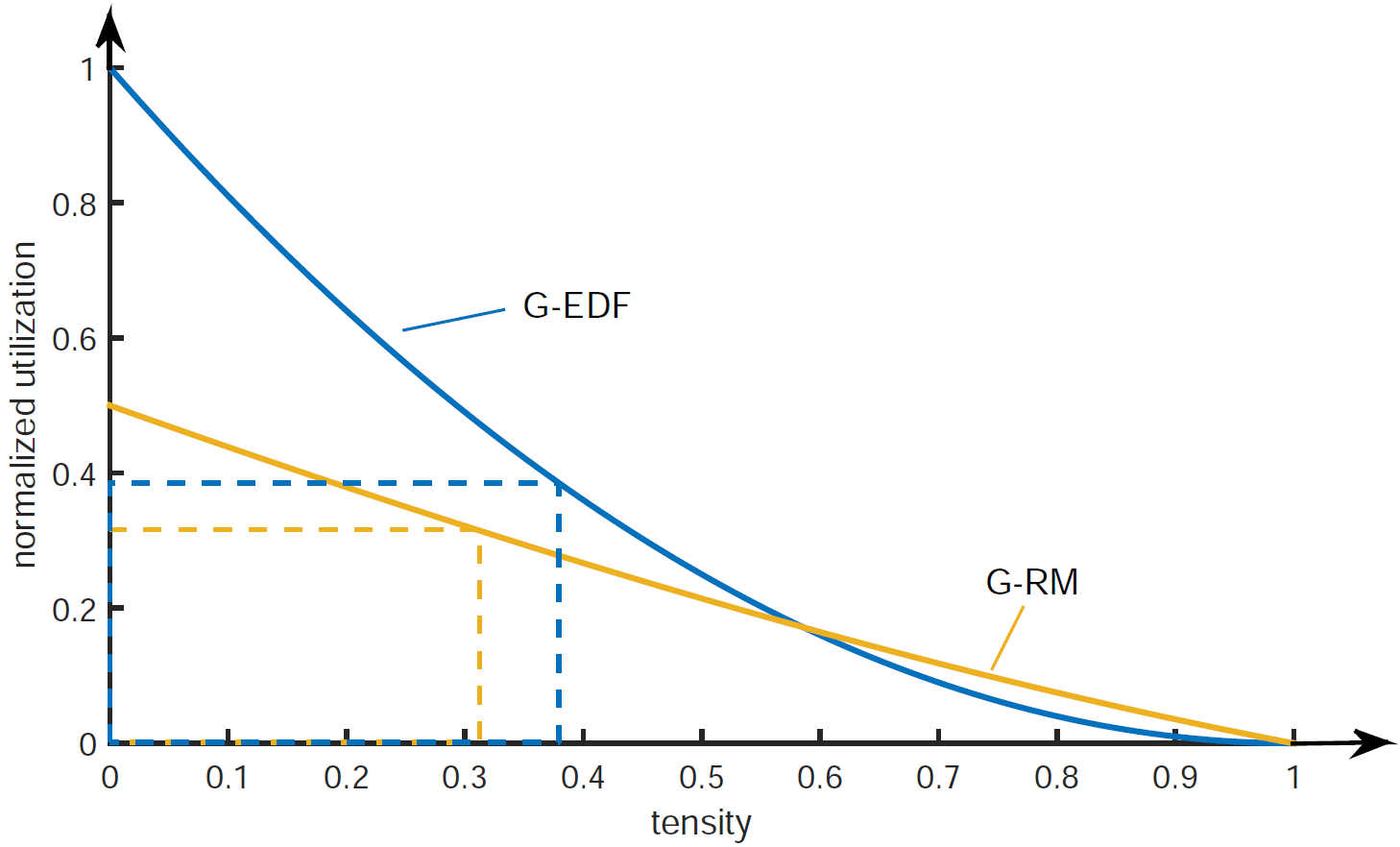} 
	\caption{Utilization-tensity bounds (solid lines) and capacity augmentation bounds (dash boxes) for G-RM and G-EDF.} \label{f:bounds}
\end{figure}

The difference between \emph{capacity augmentation bound} and \emph{utilization-tensity bound} is that the first one tests schedulability of a DAG task system by comparing both the utilization and tensity with a \emph{unified} threshold, while the second one compares utilization and tensity with its own threshold separately. Capacity augmentation bound is a little bit more abstract and simply represented by a number, so comparison of different capacity augmentation bounds is straightforward. The tensity-utilization bound is more accurate and provides more information about how each of utilization and tensity individually influences system schedulability. When directly used as schedulability test conditions, both capacity augmentation bound and utilization-tensity bound enjoy high efficiency and the 
nice property that no detailed DAG structure information is required (apart from the two abstract metrics utilization and tensity). This property makes the two bounds particularly suitable in design scenarios in which the DAG structure is unknown (e.g., in early-phase system design) or may change at run-time (e.g., for conditional DAG tasks where the actual workload released at run-time is input-dependent).

In this paper, we develop new analysis techniques for DAG task systems under Global Rate-Monotonic (G-RM) scheduling, which yields the following new schedulability bound results:

\begin{itemize}
	\item We derive the first utilization-tensity bound for G-RM scheduling. Previous work developed a utilization-tensity bound for Global Earliest-Deadline-First (G-EDF) scheduling, but no such result exists for G-RM. From graphical representations of their utilization-tensity bounds (Fig. \ref{f:bounds}), we can clearly see that while in general G-EDF has better schedulability than G-RM, for task systems with large tensity G-RM is actually superior to G-EDF.
	
	\item We improve the capacity bound of G-RM scheduling from the state-of-the-art value 3.73 \cite{li2014analysis} to 3.18.
	
\end{itemize}

We conduct experiments with randomly generated tasks, which show that our new results consistently outperform the state-of-the-art with a significant margin under different parameter settings.

\section{Model}
\label{sec:notation}

%
We consider a task set $\taskset$ that consists of $n$ sporadic tasks $\taskset=\{\tau_1, \tau_2,..., \tau_n\}$ to be executed on $m$ identical processors. Each task $\tau_i$ has a workload structure modeled by a Directed Acyclic Graph (DAG) $G_i = \langle V_i, E_i \rangle$, where $V_i$ is the set of vertices and $E_i$ is the set of edges in $G_i$. Each vertex $v  \in V_i$ is characterized by a worst-case execution time (WCET) $c(v)$. 
Each edge $(u, v) \in E_i$ represents the precedence relation between  $u$ and $v$, where $u$ is an \emph{immediate predecessor} of $v$, and $v$ is an \emph{immediate successor} of $u$.

A \emph{path} in $G_i$ is a sequence of vertices $\pi=\{v_1,v_2,...v_p\} $, where $v_{j}$ is an immediate predecessor of $v_{j+1}$ for each pair of consecutive elements $v_{j}$ and $v_{j+1}$ in $\pi$. We assume each DAG has a unique head vertex (with no predecessor) and a unique tail vertex (with no successor). This assumption does not limit the expressiveness of our model since one can always add a dummy head/tail vertex to a DAG having multiple entry/exit points. A \emph{complete path} in $G_i$ is a path $\pi$ where the first element in $\pi$ is the head vertex of $G_i$ and the last element in $\pi$ is the tail vertex in $G_i$. The length of $\pi$ is $len(\pi) =  \sum_{v \in \pi} c(v)$.   

 The \emph{volume} of $\tau_i$ is the total WCETs of all vertices of $\tau_i$:
 
 \[
 \sumC_i=\sum_{v\in V_i}c(v).
 \]
 
The \emph{critical path length} is the longest length among all paths in $G_i$:

\[
\critical_i = \max_{\forall \pi \in G_i} \{len(\pi) \}.
\]

Clearly, we have $\critical_i \leq \sumC_i$.

At run time, task $\tau_i$ releases an infinite sequence of jobs which inherit $\tau_i$’s DAG structure $G_i$. The minimum separation between the release times of two successive jobs is $T_i$. In this paper, we consider tasks with implicit deadlines, i.e., each task $\tau_i$ has a \emph{relative deadline} $T_i$. Let $J_i$ be a job of $\tau_i$, denoted as $J_i \in \tau_i$, then $r(J_i)$ and $f(J_i)$ denote $J_i$'s \emph{release time} and \emph{finish time}, respectively. $J_i$ must be finished before its \emph{absolute deadline} $d(J_i)=r(J_i)+T_i$. We call the time interval $[r(J_i),d(J_i)]$ the \textit{scheduling window} of $J_i$, whose length equals the relative deadline (i.e., period) $T_i$ of task $\tau_i$. A vertex $u$ of $J_i$ is \textit{eligible} at some time point if all its predecessors of the same job $J_i$ have finished their execution. 

The \emph{utilization} of task $\tau_i$ is defined as:
\[
u_i=\frac{C_i}{T_i}.
\]

Moreover, the \emph{total utilization} of the task set $\taskset$ is denoted as $U_{\sum} = \sum_{\tau_i\in \taskset} u_i$, and the \emph{normalized utilization} of $\taskset$ is defined as $U=U_{\sum}/m$. 

The \emph{tensity} of $\tau_i$ is defined as:

\[
\gamma_i=\frac{\critical_i}{T_i}.
\]

Moreover, the \emph{maximum tensity} among all tasks in the system is $\gamma_{max} = \max_{\tau_i \in \taskset}\{ \gamma_i \}$.

\begin{figure} [thb]
	\centering
	\includegraphics[height =4cm]{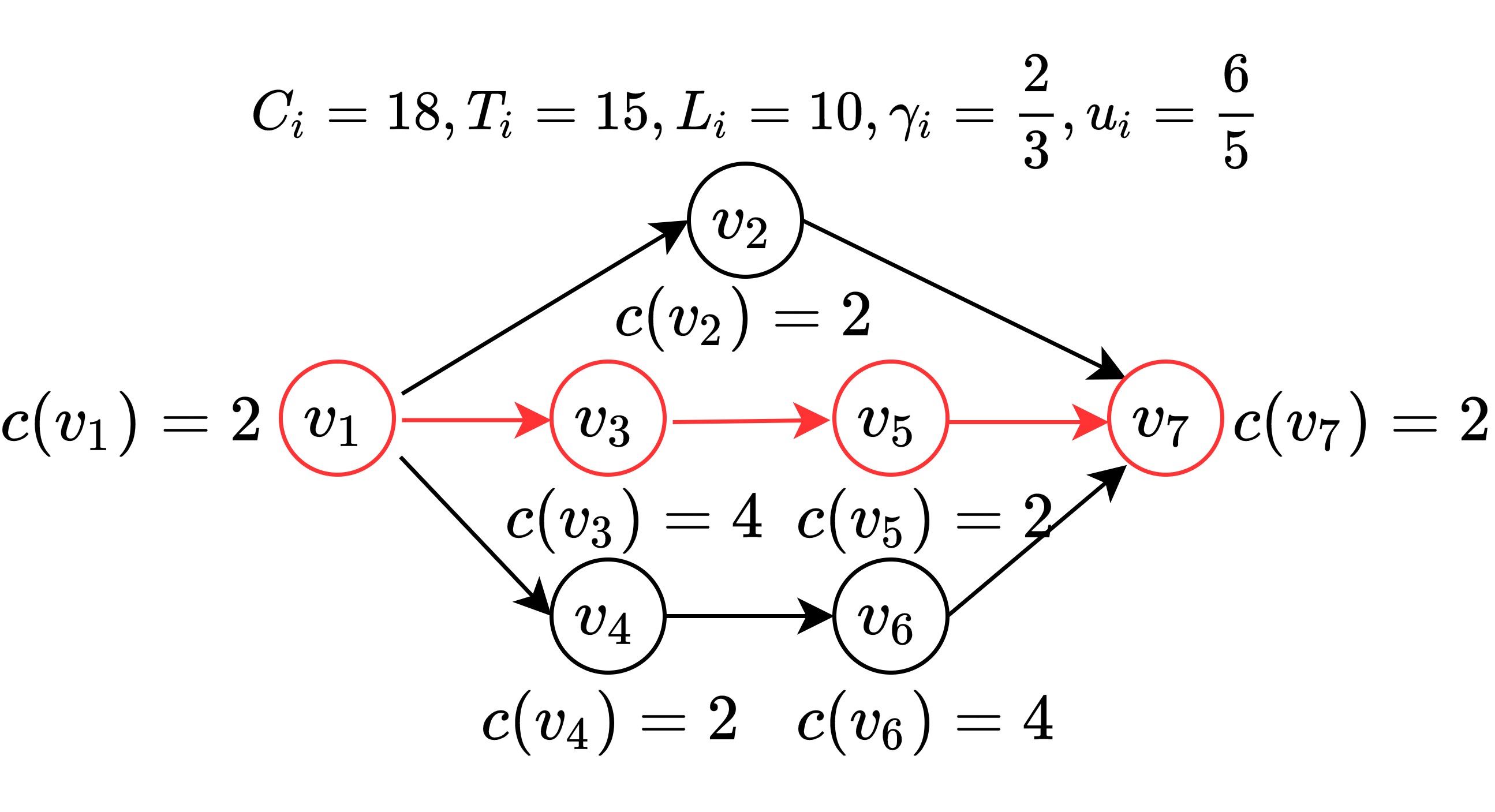} 
	\caption{An example DAG task $\tau_i$ with volume $C_i = 18$ and critical path length $L_i = 10$ (\revise{the critical path is marked in red}).}  \label{f:g1}
\end{figure} 

 For example, in Fig.\ref{f:g1}, the volume of $\tau_i$ is $C_i = 18$, and the utilization of $\tau_i$ is $u_i = \frac{6}{5}$. The critical path (marked in red)  starts from vertex $v_1$, goes through $v_3$, $v_5$ and ends at vertex $v_7$, so the critical path length of $\tau_i$'s DAG is $L_i = 2+4+2+2 = 10$. The tensity of $\tau_i$ is $\gamma_i = \frac{2}{3}$.
    
\subsection{Runtime Scheduling and Schedulability} 
 
The task set is scheduled by global priority-based scheduling algorithms on $m$ identical unit-speed processing processors. Preemption and migrations are both permitted. In this paper, we focus on Global Rate Monotonic (G-RM) scheduling. At any time, the G-RM scheduler processes the $m$ jobs with minimum period which are currently available. For tasks with implicit deadlines, G-RM scheduler is equivalent to Global Deadline Monotonic (G-DM) scheduler since $d(J_i)-r(J_i)=T_i$ holds for each job $J_i$ released by $\tau_i$.

 Without loss of generality, we assume the task system 
starts at time $0$ (i.e., the first job of the system is released at time $0$) and the time is discrete. The task set is schedulable if 
all jobs released by all tasks in $\taskset$ meet their deadlines.
Two necessary conditions must be satisfied for a task set to be schedulable\cite{li2013outstanding}:
%

\begin{lemma}
	\label{l:necessary}
 A task set $\tau$ is not schedulable (by any scheduler) unless the following conditions hold:
 \begin{align}
 		& \forall \tau_i \in \taskset : L_{i}  \leq T_i \label{eq:L<D} \\
 		&		U_{\sum}  \leq m \label{eq:U<m}
 \end{align}
 
%
%
%
%
\end{lemma}

Cleary, if (\ref{eq:L<D}) is violated for some task, then
its period is doomed to be violated in the worst case, even if it is executed exclusively on sufficiently many processors. If (\ref{eq:U<m}) is violated, then in the long term the worst-case workload of the system exceeds the processing capacity provided by the platform, and thus the backlog will increase infinitely which leads to deadline misses.


%

 \section{Background and existing results}
 In this section, we introduce some concepts and existing results that will be useful in the rest of the paper. To better understand the behind intuitions, we also review the derivations of the capacity augmentation bound and utilization-tensity bound. 
 
 \subsection{Capacity Augmentation Bound}
%
%
%
%
%
%
 
%
%
%
%
%
%
%
%
%
 The capacity augmentation bound is defined as follows:
\begin{definition}[$\!$\cite{li2013outstanding}]
	\label{l:capacitybound}
	A scheduler $S$ has a \emph{capacity augmentation bound} of $\rho$ if it satisfies for any task set $\tau$:
		\begin{align*}
		 & \forall \tau_i \in \tau : L_{i} \leq T_i/\rho  ~\wedge~ U_{\sum}  \leq  m / \rho
		\label{eq: capacity test} \\
		\Rightarrow ~
		  & \tau \textrm{~is schedulable on~} m~ \textrm{unit-speed} \textrm{~processors}
		\end{align*}
\end{definition}
From the above definition, capacity augmentation bound can be directly used to decide the schedulability of a task set on unit-speed processors. 

The capacity augmentation bound can also be stated as the following lemma:
%
 \begin{lemma}[$\!$\cite{li2013outstanding}]
 	If a scheduler $S$ can schedule any task set $\tau$ on $m$ speed-$s$ processors satisfying
 	\begin{equation}	
	\forall \tau_i \in \tau : L_i \leq T_i  ~\wedge~ U_{\sum}\leq m
 	\label{eq: origincapacity}
 	\end{equation}
 	then $S$ has a capacity augmentation bound of $s$.  
 	 	\label{th:determincapa}
 \end{lemma}

From the above lemma, the capacity augmentation bound can also be used to quantify the relative performance of different scheduling approaches. Clearly, smaller capacity augmentation bound implies better schedulability. 

 Before going deep, we first introduce two useful concepts:
 
  \begin{definition}[$\!$\cite{li2014analysis}]
Given a task $\tau_i$,  $q_i(t,s)$ is the total work 
finished by $S_{\infty,s}$ on speed-$s$ processors in interval $[r_i, r_i + t]$ where $r_i$ is the release time of $\tau_i$, and $S_{\infty,s}$ is a hypothetical scheduling strategy that must schedule a task set on infinitely many speed-$s$ processors. 
  \end{definition}

   \begin{definition}[\cite{li2014analysis,bonifaci2013feasibility,baruah2014improved}] \label{def:maxLoad}
   Given a task $\tau_i$, $work(\tau_i,t,s)$ is defined by
%
   	\[
   	work(\tau_i,t,s) = 
   	\left\{\!\!
   	\begin{aligned}
   	& C_{i}-q_i(T_i-t,s), ~~~~~~~~~~ t\leq T_i \\
   	& \lfloor\frac{t}{T_i}\rfloor C_{i} \!+\! work(\tau_i,t\!-\!\lfloor\frac{t}{T_i}\rfloor T_i,s), ~ t\!>\!T_i.
   	\end{aligned}
   	\right.
   	\]
   \end{definition}

Intuitively, $work(\tau_i, t, s)$ denotes the maximum amount of
work from jobs with deadlines that fall within any interval of $I$ finished by schedule $S_{\infty,s}$ during $I$ over all job sequences that may be generated by $\tau_i$, where $|I|=t$.

 Fig.\ref{fig:work} illustrates the execution sequences of task $\tau_i$ in Fig. \ref{f:g1} scheduled by $S_{\infty,1}$ and $S_{\infty,2}$, from which we can see $q_i(2,1)=2$ and $q_i(2,2)=2\times1+2\times3=8$, and thus $work(\tau_i,13,1)=18-2=16$ and $work(\tau_i,13,2)=18-8=10$.

The best known results of capacity augmentation bounds for G-EDF and G-RM are derived indirectly by examining the schedulability of $\tau$ on speed-$s$. The related useful results are stated as follows.


  \begin{figure}
  	\centering
	\subfigure[$\tau_i$ on speed-1 processors]{\includegraphics[width=3.8in]{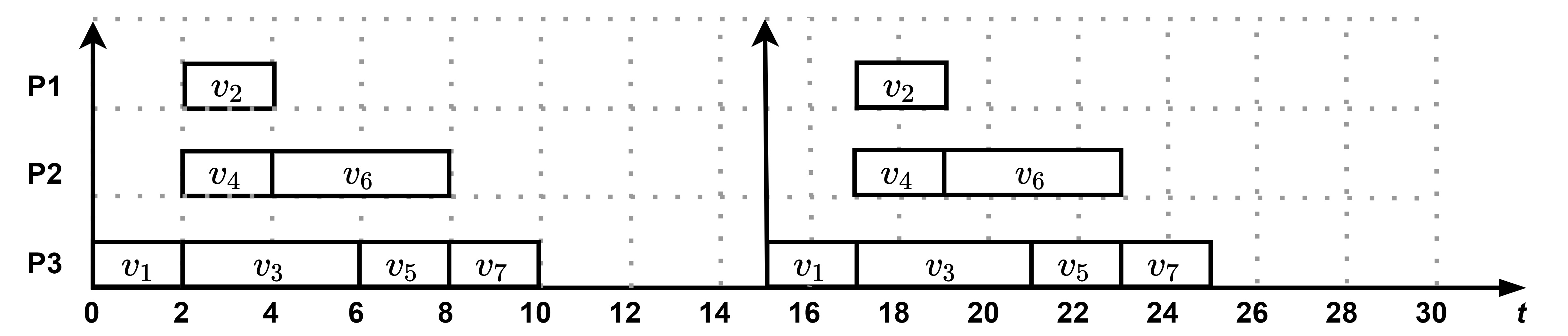}}
  	\hspace{0.02in}
	\subfigure[$\tau_i$ on speed-2 processors]{\includegraphics[width=3.8in]{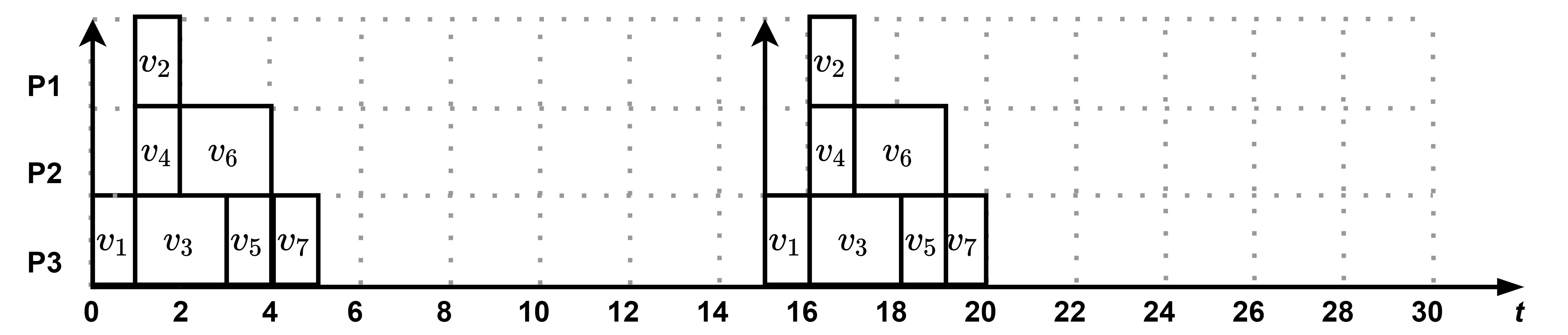}}
  	\caption{The execution of task $\tau_k$ on processors with different speeds.}
  	\label{fig:work}
  \end{figure}

 \begin{lemma}[\cite{li2014analysis}] \label{lem:Li}
 	Given any task $\tau_i$,  $\forall t>0$ and $s>1$, 
 	\[
 	\frac{work(\tau_i,t,s)}{t}\leq 
 	\left\{
 	\begin{aligned}
 	& \frac{u_{i}-1}{1-\frac{1}{s}}, &   u_i > s \\
 	& u_i,    &   0 \leq u_{i} \leq s.
 	\end{aligned}
 	\right.
 	\]
 \end{lemma} 
%
 \begin{lemma}[\cite{li2014analysis}]\label{lem:LiGEDF}
 	A task set $\tau$ is schedulable by GEDF on $m$ speed-$s$ processors ($s > 1$) if $\forall \tau_i \in \tau : L_{i} \leq T_i$, and
 	\begin{equation}
 	\forall t>0: \sum_{\forall \tau_i \in \tau} work(\tau_i, t, s)\leq (s \times m-(m-1))\times t.
 	\label{eq: Li-speed}
 	\end{equation}	
 \end{lemma}

 \begin{lemma}[\cite{li2014analysis}]\label{lem:LIGRM}
	A task set $\tau$ is schedulable by G-RM on $m$ speed-$s$ processors ($s > 1$) if $\forall \tau_i \in \tau : L_{i} \leq T_i$, and
	\begin{equation}
	\forall t>0: \sum_{\forall \tau_i \in \taskset} work(\tau_i, t, s)\leq \frac{(s \times m-(m-1))\times t}{2}.
	\label{eq: Li-speed2}
	\end{equation}	
\end{lemma}

By combining Lemma \ref{lem:Li} with \ref{lem:LiGEDF} and \ref{lem:LIGRM} respectively, we can get

 \begin{lemma}
 Any task set $\tau$ satisfying $\forall \tau_i \in \tau : L_{i} \leq T_i$, and $U_{\sum}\leq m$ is schedulable under G-EDF and G-RM on $m$ processors of speed- $\frac{3+\sqrt{5}}{2}$ and $2+\sqrt{3}$ respectively.
\label{th:GEDFBOUND}	
 \end{lemma}
The above lemma implies capacity augmentation bounds of $\frac{3+\sqrt{5}}{2}$ and  $2+\sqrt{3}$ for G-EDF and G-RM(on unit-speed processors) respectively.

\subsection{Utilization-Tensity Bound}

The drawback of capacity augmentation bound is that it gives the same threshold for both 
normalized utilization and tensity, i.e., a task set with small normalized
utilization but maximum tensity slightly exceeding capacity augmentation bound cannot pass the schedulability test even it is
easy to be scheduled. From this view, another metric named utilization-tensity bound which aims to have asymmetric thresholds for normalized utilization and tensity on unit-speed processors are derived in \cite{xuutilizationtensity}. Instead of deriving schedulability conditions on speed-$s$ processors and then transforming them to speed-$1$ processors, conditions to guarantee schedulability on unit-speed processors for G-EDF are directly derived by obtaining the \enquote{unit-speed version} of Lemma \ref{lem:Li} and Lemma \ref{lem:LiGEDF}:

\begin{lemma}[\cite{xuutilizationtensity}]\label{lem:pi}
	For any  task $\tau_i$ and $t>0$, it holds
	\begin{equation}
	\frac{work(\tau_i, t,1)}{t} \leq \left\{
	\begin{aligned}
	& \frac{u_{i}-\gamma_{i}}{1-\gamma_{i}}, &  u_{i}> 1 \\
	& u_{i},    &  0\leq u_{i} \leq 1.
	\end{aligned}
	\right.
	\end{equation}	
\end{lemma}

\begin{lemma}[\cite{xuutilizationtensity}]\label{lem:7}
	A task set $\tau$ is schedulable by G-EDF on $m$ unit-speed processors if 
	\begin{equation}\label{eq:7-1}
	\forall t>0: \sum_{\tau_i \in \tau} work(\tau_i, t,1) \leq (m  - \gamma_{max} (m-1)) t 
	\end{equation}
	where $\displaystyle \gamma_{max} = \max_{\tau_i \in \tau} \{\gamma_{i}\}$.
\end{lemma}

At last, by combining Lemma \ref{lem:pi} and Lemma \ref{lem:7}, asymmetric thresholds for normalized utilization and tensity, i.e., the utilization-tensity bound, to guarantee schedulability of G-EDF on unit-speed machines are derived.

\begin{lemma} [\cite{xuutilizationtensity}]\label{GEDF:utilization-tensity}
	A task set t is schedulable under G-EDF on $m$ unit-speed processors if
	\[
	\forall \tau_i \in \tau : L_i \leq T_i  ~\wedge~ U\leq (1-\gamma_{max})^2
	\]
	where $\displaystyle \gamma_{max} = \max_{\tau_i \in \tau} \{\gamma_{i}\}$.
\end{lemma}

\section{Straightforward Results for G-RM}
In this section, we first extend existing techniques introduced in the last section to obtain a utilization-tensity bound for G-RM, which is further shown to be pessimistic due to its ignorance of some key properties produced by the scheduling behavior of G-RM. Then in the next section we develop new techniques and present tighter utilization-tensity bound and capacity augmentation bound for G-RM. 

It is nature to imply the techniques introduced in \cite{xuutilizationtensity}, to obtain a unit-speed version of Lemma \ref{lem:LIGRM}:

\begin{lemma}\label{lem:8}
	A task set $\tau$ is schedulable by G-RM on $m$ speed-$1$ processors if 
	\begin{equation}\label{eq:8-1}
	\forall t>0: \sum_{\tau_i \in \tau} work(\tau_i, t,1) \leq \frac{ (m  - \gamma_{max} (m-1)) t }{2}
	\end{equation}
	where $\displaystyle \gamma_{max} = \max_{\tau_i \in \tau} \{\gamma_{i}\}$.
\end{lemma}

\begin{proof}
The proof is omitted here. The process to obtain the "unit-speed version" of Lemma \ref{lem:LIGRM} is similar with transforming the
upper bound of $work(\tau_i, t, s)$ in Lemma \ref{lem:Li} to an upper bound of $ work(\tau_i, t,1)$ in Lemma \ref{lem:7} for G-EDF \footnote{The transforming is through a scaling technique. Refer to \cite{xuutilizationtensity} for more details.}. 
\end{proof}

Then by combining Lemma \ref{lem:pi} and Lemma \ref{lem:8}, we can derive a utilization-tensity bound for G-RM. 

\begin{lemma} \label{l:basicbound}
	A task set $\taskset$ is schedulable under G-RM on $m$ unit-speed processors if
	\begin{equation} \label{eq:basicbound}
	\forall \tau_i \in \tau : L_i \leq T_i  ~\wedge~ U\leq \frac{ (1-\gamma_{max})^2}{2} 
    \end{equation}
	where $\displaystyle \gamma_{max} = \max_{\tau_i \in \tau} \{\gamma_{i}\}$.
\end{lemma}

\begin{proof}
If condition (\ref{eq:basicbound}) is satisfied, then we have 
\begin{eqnarray*}
	&&U_{\sum} \leq\frac{(1-\gamma_{max})^2 m +\gamma_{max}(1-\gamma_{max})}{2} \\
	&\Leftrightarrow&  \frac{U_{\sum}}{1-\gamma_{max}} \leq \frac{ (m  - \gamma_{max} (m-1)) t }{2}.
\end{eqnarray*}

From Lemma \ref{lem:pi} we have :

\begin{align*}
&\sum_{\forall \tau_i \in \tau, u_i >1}work(\tau_i, t, 1) \leq \sum_{\forall \tau_i \in \tau, u_i >1}\!\!\! \frac{u_i-\gamma_i}{1-\gamma_i}+\sum_{\forall \tau_i \in \tau, u_i \leq 1}\!\!\! u_i  \\
&\leq \frac{\displaystyle\sum_{\forall \tau_i \in \tau, u_i >1}\!\!\!u_i+\!\!\!\! \sum_{\forall \tau_i \in \tau, u_i \leq1}\!\!\!u_i-(\sum_{\forall \tau_i \in \tau, u_i >1}\!\!\!\gamma_i+\gamma_{max}\!\!\!\!\sum_{\forall \tau_i \in \tau, u_i \leq 1}u_i)}{1-\gamma_{max}} \\
&\leq\frac{U_{\sum}}{1-\gamma_{max}}
\end{align*}

Then from Lemma \ref{lem:8}, we know $\tau$ is schedulable under G-RM on $m$ unit-speed processors. 
\end{proof}

Moreover, since condition (\ref{eq:basicbound}) always meets if $U\leq\frac{1}{2+\sqrt{3}}$ and $\gamma_{max} \leq \frac{1}{2+\sqrt{3}}$, lemma \ref{l:basicbound} implies the same capacity augmentation bound of $2+\sqrt{3}$ for G-RM with the best result of the-state-of-the-art. 

\section{New analysis for G-RM}
In the following, we develop new techniques for analyzing the schedulablity of DAG task systems under G-RM. 

We first give some results that will be useful in our analysis. 

\begin{lemma} \label{l: schedulablework}
	If the total interfering workload $\interfering_i$ on a job $J_i$ is bounded by
	$\interfering_i \leq m T_i -(m-1)L_i$, then job $J_i$ can meet its deadline on $m$ processors under G-RM.
\end{lemma}

\begin{proof}
We prove this lemma by contradiction. Suppose that $\interfering_i \leq m T_i -(m-1)L_i$ and $J_i$ misses its deadline. Let $I=[r(J_i), d(J_i))$, i.e., the scheduling window of $J_i$. Let an in-complete interval denote a time interval where at least one processor is idle at any time in this interval. Since for each time unit of in-complete interval, the remaining critical path of $J_i$ is reduced by one time unit. Then we know the total length of in-complete intervals during $I$ is no more than $L_i$, otherwise $J_i$ must finish its execution. We denote by $X$ the total length of the
intervals within $I$ where in the G-RM schedule all $m$ processors
are busy. Then we have the total work done by G-RM during $I$ is at least:
\begin{eqnarray*}
&& m*X+(T_i-X)=(m-1)X+T_i \nonumber \\
&\geq& (m-1)(T_i-L_i)+T_k=mT_i-(m-1)L_i. 
\end{eqnarray*}

Since $\interfering_i \leq m T_i -(m-1)L_i$, then $J_i$ must be finished at its deadline, reaching a contradiction. 
\end{proof}

Let $\jobset$ denote the minimum set of jobs released from $\taskset$ and failing to be schedulable under G-RM. Suppose that $J_k$ from $\tau_k$ is the first job that misses its deadline. Without loss of generality, we assume that there are no jobs with period greater than $T_k$ in $\jobset$, since the removal of such jobs from $\jobset$ does not affect G-RM, i.e., $J_k$ will still miss its deadline ($J_k$ will only be interfered by other jobs from tasks with periods no greater than $T_k$ under G-RM). A straightforward property can be observed. 

\begin{lemma} \label{sinfty}
	$S_{\infty,s_1}$ finishes at least as every vertex of each job in $\jobset$ than $S_{\infty,s_2}$ if $s_1\geq s_2$ during a interval of $[0, t)$.
\end{lemma} 

\begin{proof}
	This can be proved by induction. Clearly, the lemma is true for time $t=0$. Suppose that the lemma is true at $t*$. Then during the next time unit, every vertex of each jobs processed by $S_{\infty,s_2}$ that is unfinished by $S_{\infty,s_1}$ is eligible for $S_{\infty,s_1}$. Since a processor with speed $s_1$ can process no less workload than a processor with speed $s_2$ during each time unit, the lemma is true. 
\end{proof}

 Let $\work^{\jobset, I}$ denote the work of $\jobset$ done by $S_{\infty,1}$ during an interval of $I$ and $A^{\jobset, I}$ denote the work of $\jobset$ done by the scheduler of G-RM during the interval of $I$. The following lemma gives us a necessary condition for task set $\tau$ to be unschedulable by G-RM.

\begin{lemma}\label{l: keylemma}
	If a task set $\tau$ satisfying necessity conditions (\ref{eq:L<D}) and (\ref{eq:U<m}) is not schedulable under G-RM on $m$ unit-speed processors, then there must be a task $\tau_k \in \tau$ and a time interval $I=[t, d(J_k))$, where $|I|\geq T_k$ and:
	\[
	\work^{\jobset, I}> (m-(m-1)\gamma_{k}) |I|.
	\]
\end{lemma}

\begin{proof}
	Suppose that $J_k$ from $\tau_k$ be the first job that misses its deadline. Let $I^{'}=[0, r(J_k))$. We prove the lemma by distinguishing two cases:
	\begin{itemize}
		\item \textbf{Case 1}: $\work^{\jobset, I^{'}} \leq A^{\jobset, I^{'}}$. Sine  $S_{\infty,1}$ is feasible and G-RM fails at $d(J_k)$, we know that the work of $\jobset$ finished by $S_{\infty,1}$ during $[r(J_k), d(J_k))$ must be more than the work done by G-RM during $[r_k, d_k)$, i.e., $\work^{\jobset, [r(J_k), d(J_k))}> A^{\jobset, [r(J_k), d(J_k))}$. Since $J_k$ misses its deadline at $d(J_k)$, from Lemma \ref{l: schedulablework}, we have 
		\begin{equation} \label{eq:temp1}
		\work^{\jobset, [r(J_k), d(J_k))}\geq  (m-(m-1)\gamma_k) T_k
		\end{equation} 
		
		\item \textbf{Case 2}:  $\work^{\jobset, I^{'}} > A^{\jobset, I^{'}}$. Let $t^* \leq r(J_k)$ denote the latest point in time such that at any time $t \in [0, t*)$ the scheduler of G-RM has processed at least as much of every vertex of each job in $\jobset$ as $S_{\infty,1}$ at time $t$. Such a time exists, since $t^{*} = 0$ satisfies this property. Let $I^{*}=[t^*, d(J_k))$. Then we have 
		\[
		\work^{\jobset, [0, t^*)} \leq A^{\jobset, [0, t^*)}
		\]
		Therefore
		\[
		\work^{\jobset, I^{*}} > A^{\jobset,I^{*}}.
		\] 
		Then it is sufficient to prove this lemma by proving that G-RM finishes more than  $(m-(m-1)\gamma_k) |I^{*}|$ units of work within $I^{*}=[t^*, d(J_k))$.
		
	   Since $\gamma_k \leq 1$,	from Lemma \ref{sinfty}, $S_{\infty,1}$ has processed at least as much as every vertex of each job in $\jobset$ as $S_{\infty, \gamma_k}$ at any time $t \in [0, t*)$. Then we know the scheduler of G-RM has processed at least as much of every vertex of each job in $\jobset$ as $S_{\infty, \gamma_k}$ at any time $t \in [0, t*)$.     
		
		 Let $X$ denote the total length of the intervals within $I^{*}$ where all $m$ processors are busy in the G-RM schedule, and $Y=|I^{*}|-X$, i.e., the total length of intervals within $I^{*}$ where at least one processor is idle. We distinguish two cases. First assume that $Y\geq \gamma_{k}|I^{*}|$. Let $\Theta_1, \cdots, \Theta_\varphi \subseteq I $ denote all sub-intervals of $I^{*}$ where not all processors are busy.

		
		Let $\eta=\frac{\gamma_k(d(J_k)-t^{*})}{\delta}=\frac{\gamma_k|I^{*}|}{\delta}$, where $\eta$ is a positive integer and $\delta$ is an small positive number\footnote{$\delta$ can be considered as the length of the minimum time unit, which is small enough that is divisible by $\gamma_k(d(J_k)-t^{*})$.}.

		We define two sequences of time points: $t_0, t_1, \cdots, t_{\eta}$ and $t_0^{'}, t_1^{'}, \cdots, t_{\eta}^{'}$, where $t_i=t^{*}+i*\frac{\delta}{\gamma_k}$, $t_{0}^{'}=t^{*}$ and $|[t*, t_{i}^{'})\cap \cup_{i}\Theta_i|=\delta\times i$. Then we have $t_{\eta}=t^*+|I^*|=d(J_k)$, and $|[t*, t_{\eta}^{'})\cap \cup_{i}\Theta_i|=\gamma_k|I^*|$. Since $|\cup_{i}\Theta_i|\geq \gamma_k|I^*|$, we know $t_{\eta}^{'} \leq d(J_k)$. 
		
		Then we prove that by $t_{\eta}^{'}$ G-RM has finished as much as every vertex of each job as $S_{\infty, \gamma_k}$ by $t_{\eta}$. The proof is by induction. Clearly, from the definition of $t^{*}$, G-RM has finished as much as every vertex of each job as $S_{\infty, \gamma_k}$ by $t^{*}$. Then suppose that by $t_i^{'}$ G-RM has finished as much as every vertex of each job as $S_{\infty, \gamma_k}$ by $t_i$ where $0<i<\eta$. Then at each time point
		during $|[t_i^{'}, t_{i+1}^{'})\cap \cup_{i}\Theta_i|$ all vertices of each job that are unfinished by G-RM and processed by $S_{\infty, \gamma_k}$ during $[t_i, t_{i+1})$ are available for G-RM. Since
		during all these time points G-RM does not use all processors, by time $t_{i+1}^{'}$, it has processed at least as much of every job
		as $S_{\infty, \gamma_k}$ by time $t_{i+1}$. Hence at time $t_{\eta}^{'}$ G-RM has finished as much of every vertex of each job as $S_{\infty, \gamma_k}$ by $t_{\eta}=d(J_k)$.
			
		
		Since $L_k/\gamma_k = T_k$, $J_k$ can meet its deadline under $S_{\infty, \gamma_k}$. This contradicts with the assumption that $J_k$ misses its deadline under G-RM.
		
		
		
		Then we know $Y<\gamma_{k}|I^{*}|$. Therefore, the work of $\jobset$ that G-RM finishes during $I^{*}$ is at least
		\begin{eqnarray*}
		mX+Y&=&m(|I^{*}|-Y)+Y=m I^{*}-(m-1)Y \nonumber \\
		    &>&m |I^{*}|-(m-1)\gamma_{k}|I^{*}| \nonumber \\
		    &=&(m -(m-1)\gamma_{k})|I^{*}|
		\end{eqnarray*}
		Therefore, it holds
		\begin{equation} \label{eq:temp2}
		\work^{\jobset, I^{*}} > A^{\jobset,I^{*}}\geq (m -(m-1)\gamma_{k})|I^{*}|
		\end{equation} 		
	
\end{itemize}
	Then by construction of $I$ in the above two cases, i.e., $[r(J_k), d(J_k))$ in (\ref{eq:temp1}) and $I^*$ in (\ref{eq:temp2}), the lemma is proved. 	
\end{proof}	

\subsection{utilization-tensity bound}
In the following we present analysis to derive a much tighter utilization-tensity bound for G-RM than condition (\ref{eq:basicbound}). 

From Lemma \ref{l: keylemma}, we can identify a sufficient condition for a task set to be schedulable under G-RM by violating the necessary condition pointed out in Lemma \ref{l: keylemma}. Clearly, in order to find such a necessary condition for task set $\tau$ to be unschedulable by G-RM, we need to bound $\work^{\jobset, I}$.


\begin{lemma}\label{l:maximumwork}
 For any interval of $I=[t, d(J_k))$, where $d(J_k)$ is the deadline of a job $J_k$ from $\tau_k$ and $|I|\geq T_k$, i.e., $t<r(J_k)$, it is satisfied:
 \[
 \work^{\jobset, I}	\leq \sum_{\forall \tau_i \in \tau, T_i \leq T_k}work(\tau_i,|I|,1)+\sum_{\forall i, i \neq k \wedge T_i\leq T_k} C_i
 \]
 
\end{lemma}
	
\begin{proof}
	
We consider two cases:
\begin{itemize}
	\item $\tau_k$. The work of jobs in $\jobset$ from $\tau_k$ done by $S_{\infty,1}$ during the interval of $I$ is upper bounded by $work(\tau_k,|I|,1)$. 
	\item $\forall i, i \neq k \wedge T_i\leq T_k$. For each $\tau_i$, we divide $I$ into two parts: $I_1=[t, d_i^{*})$ and $I_2=[d_i^{*}, d_k)$ where $d_i^{*}$ denote the latest job from $\tau_i$ with deadline falling within $I$. Then the work of jobs in $\jobset$ from $\tau_i$ done by $S_{\infty,1}$ during $I$ is no more than $work(\tau_i,|I_1|,1)+C_i\leq work(\tau_i,|I|,1)+C_i$. 
\end{itemize}
	
By summing up the work of all jobs in $\jobset$ done by $S_{\infty,1}$ in $I$, the lemma is proved.	
\end{proof}

With Lemma \ref{l: keylemma} and Lemma \ref{l:maximumwork}, we are now ready to derive a sufficient condition for G-RM to be schedulable:

\begin{lemma}\label{l:connectworkt}
	A task set $\tau$ satisfying necessity conditions (\ref{eq:L<D}) and (\ref{eq:U<m}) is schedulable by G-RM on $m$ unit-speed processors if for $\forall k, \tau_k \in \taskset$, it satisfies:
	
	\begin{equation}\label{eq:RMugly}
	\sum_{\forall \tau_i \in \tau \atop T_i \leq T_k}\mathop{\rm{sup}}\limits_{t\geq T_i} \frac{\lambda_{i,k}(t)}{t}\leq m  - \gamma_{k} (m-1)
	\end{equation}	

   where $\lambda_{i,k}(t)=work(\tau_i, t, 1)+C_i$ if $i\neq k$, otherwise $\lambda_{i,k}(t)=work(\tau_i, t, 1)$.
\end{lemma}

\begin{proof}
	We prove the lemma by contradiction. Suppose $\forall i, \tau_i \in \tau$, (\ref{eq:RMugly}) is met and $\tau$ is not schedulable by G-RM. Then from Lemma \ref{l: keylemma} and Lemma \ref{l:maximumwork}, we know there must exist a task $\tau_k$ and a time interval $I$, where $|I|\geq T_k$ and 
	\[
	(m-(m-1)\gamma_{k}) |I|< \sum_{\forall \tau_i \in \tau \atop T_i \leq T_k}work(\tau_i,|I|,1)+\sum_{\forall i, i \neq k \wedge T_i\leq T_k} C_i,
	\]
	
   which implies that 
	
		\[
	m-(m-1)\gamma_{k}< \sum_{\forall \tau_i \in \tau \atop T_i \leq T_k} \frac{\lambda_{i,k}(|I|)}{|I|}.
	\]

	Reaching a contradiction with (\ref{eq:RMugly}), the lemma is proved. 
\end{proof}

In the following, we show that it is not necessary to enumerate each instance of $t\geq T_i$ to check (\ref{eq:RMugly}). 

\begin{lemma}\label{uglyreduce}
$\mathop{\rm{sup}}\limits_{t\geq T_i}\frac{work(\tau_i, t, 1)+\Delta}{t} = \mathop{\rm{sup}}\limits_{2T_i  \geq t\geq T_i} \frac{work(\tau_i, t, 1)+\Delta}{t}$.
where $\Delta \geq 0$.
\end{lemma}

\begin{proof}

To prove the lemma, it is sufficient to prove 
\begin{equation}\label{eq:temp}
\mathop{\rm{sup}}\limits_{t\geq 2T_i}\frac{work(\tau_i, t, 1)+\Delta}{t} \leq \mathop{\rm{sup}}\limits_{2T_i  \geq t\geq T_i} \frac{work(\tau_i, t, 1)+\Delta}{t}. 
\end{equation}

Sine $\frac{work(\tau_i, T_i, 1)+\Delta}{T_i}=u_i+\frac{\Delta}{T_i}$, we have $\mathop{\rm{sup}}\limits_{t\geq T_i}\frac{work(\tau_i, t, 1)+\Delta}{t} \geq  \mathop{\rm{sup}}\limits_{2T_i  \geq t\geq T_i} \frac{work(\tau_i, t, 1)+\Delta}{t} \geq u_i+\frac{\Delta}{T_i}$. Let $t=aT_i+b$, where $b$ is a positive number and $b\in [0, T_i)$, $a$ is an integer and $a\geq 1$. If $\frac{work(\tau_i, b, 1)}{b} < u_i+\frac{\Delta}{T_i}$, then we have 
\begin{eqnarray*}
&&\frac{work(\tau_i, t, 1)+\Delta}{t}=\frac{aC_i+work(\tau_i, b, 1)+\Delta}{aT_i+b} \\
&<&\frac{(aT_i+b)(u_i+\Delta/T_i)}{aT_i+b}=u_i+\frac{\Delta}{T_i} 
\end{eqnarray*}

  Then it is sufficient to prove (\ref{eq:temp}) by proving that for each instance $b\in [0, T_i)$ where $\frac{work(\tau_i, b, 1)}{b} \geq u_i+\frac{\Delta}{T_i}$, it satisfies:
\[
\frac{work(\tau_i, aT_i+b, 1)+\Delta}{aT_i+b} \leq \frac{work(\tau_i, T_i+b, 1)+\Delta}{T_i+b}
\]

where $a \geq 2$. Clearly, we have
\begin{eqnarray*}
&&\frac{work(\tau_i, T_i+b, 1)+\Delta}{T_i+b}-\frac{work(\tau_i, aT_i+b, 1)+\Delta}{aT_i+b}\\
&=&\frac{C_i+work(\tau_i, b, 1)+\Delta}{T_i+b}-\frac{aC_i+work(\tau_i, b, 1)+\Delta}{aT_i+b} \\
&=&\frac{(a-1)(T_i(work(\tau_i, b, 1)+\Delta)-C_ib)}{(aT_i+b)(T_i+b)} \\
&\geq& 0 ~~~~[\because \frac{work(\tau_i, b, 1)}{b} \geq u_i ~~\mbox{and}~~ a\geq 2]
\end{eqnarray*}
The lemma is proved.	
\end{proof}

Lemma \ref{l:connectworkt} can be seen as a byproduct of this paper. The schedulability of $\taskset$ under G-RM can be decided by checking (\ref{eq:RMugly}) for each task $\tau_i \in \tau$ for each discrete value of $t \in [T_i, 2T_i)$ separated by the minimum time unit. Clearly, such test has a polynomial-time complexity. Nevertheless, Lemma \ref{l:connectworkt} still results in high complexity and requires the intra-structure information, which is not always acceptable. Recall that our focus in this paper is to derive quantitative schedulability bounds for G-RM.   

In the following, we present a schedulability test where no intra-structure information is required. We begin with deriving a bound of $work(\tau_i, t, 1)$ when $t\geq T_i$, which is tighter than that in Lemma \ref{lem:pi}. In particular, tasks are divided into two groups: heavy tasks whose utilization is greater than 1 and light tasks whose utilization is no greater than 1.

%

\begin{lemma}\label{lem:loadbigt}
	For any  task $\tau_i$ and $t\geq T_i$, it holds
	\begin{equation}
	\frac{work(\tau_i, t,1)}{t} \leq \left\{
	\begin{aligned}
	& \frac{2u_{i}-\gamma_{i}}{2-\gamma_{i}}, &  u_{i}> 1 \\
	& u_{i},    &  0\leq u_{i} \leq 1.
	\end{aligned}
	\right.
	\end{equation}	
\end{lemma}

\begin{proof}
	For any $t$, we split it into two parts: $t=t_1+t_2$, where $t_1=\lfloor\frac{t}{T_i}\rfloor T_i$ and $t_2=t-t_1$. Clearly, $t_1\geq T_i$ and $t_2<T_i$.
	
	According to Definition \ref{def:maxLoad}:
	
	\begin{eqnarray*}
		work(\tau_i,t)&=&C_i\lfloor\frac{t}{T_i}\rfloor+C_i-q_i(T_i-t_2,1) \\
		&=&u_it_1+C_i-q_{i}(T_i-t_2,1).
	\end{eqnarray*}

	Since $S_{\infty,1}$ finishes $C_i$ at time $L_i$, then for any $T_i\geq t>L_i$, we have $q_{k}(t) = C_i$. For any $t < L_i$, at every time
	instant in the interval $[r_i, r_i +t)$, there is at least one processor executing the workload of $\tau_i$, and thus, $q_{i}(t,1)\geq t$. Then we have:
	
	\begin{equation}
	q_{i}(t,1) \leq \left\{
	\begin{aligned}
	&t, &  t\leq L_i \\
	& C_i,    & T_i\geq t>L_i.
	\end{aligned}
	\right.
	\end{equation}
	
	Then we have:
	
	\begin{equation*}
	{work(\tau_i, t,1)} \leq \left\{
	\begin{aligned}
	& u_it_1+t_2+C_i-T_i, & t_2\geq T_i-L_i \\
	& u_{i}t_1    &  else.
	\end{aligned}
	\right.
	\end{equation*}
	
	We consider two cases:
	\begin{itemize}
		\item $u_i\leq 1$. If $t_2\geq T_i-L_i$, we have

		\begin{eqnarray*}
			\frac{work(\tau_i,t)}{t}&=&\frac{u_it_1+t_2+C_i-T_i}{t} \\
			&\leq&\frac{t+C_i-T_i}{t} \\
			&=&\frac{C_i-T_i}{t}+1 \\
			&\leq&\frac{C_i-T_i}{D_i-L_i}+1\\
			&=& \frac{u_i-\gamma_i}{1-\gamma_i}
		\end{eqnarray*}
		
		If $t_2< T_i-L_i$. We have $\frac{work(\tau_i,t)}{t}=\frac{u_it_1}{t}\leq u_i$. Since $u_i \leq 1$, $\frac{u_i-\gamma_i}{1-\gamma_i}\leq u_i$. In summary,  $\frac{work(\tau_i,t)}{t}\leq u_i$ when $u_i\leq 1$. 
		
		In summary, $\frac{work(\tau_i,t)}{t} \leq u_i$ when $u_i \leq 1$.
		
		\item $u_i> 1$.  If $t_2\geq T_i-L_i$, we have 
		
		\begin{eqnarray*}
			\frac{work(\tau_i,t)}{t}&=&\frac{u_it_1+t_2+C_i-T_i}{t} \\
			&=& \frac{u_i(t_1+T_i)-(T_i-t_2)}{(t_1+T_i)-(T_i-t_2)}\\
			&\leq& \frac{2u_iT_i-(T_i-t_2)}{2T_i-(T_i-t_2)} \\
			&\leq& \frac{2u_iT_i-L_i}{2T_i-L_i} \\
			&=& \frac{2u_i-\gamma_i}{2-\gamma_i}
		\end{eqnarray*}
		
		If $t_2< T_i-L_i$. We have $\frac{work(\tau_i,t)}{t}=\frac{u_i t_1}{t}\leq u_i$. Since $u_i > 1$, we have $\frac{2u_i-\gamma_i}{2-\gamma_i}>u_i$.
		
		In summary, $\frac{work(\tau_i,t)}{t} \leq \frac{2u_i-\gamma_i}{2-\gamma_i}$ when $u_i > 1$.	
	\end{itemize}
	In both cases, the lemma is proved. 
\end{proof}

Now we can combine Lemma \ref{l: keylemma}, Lemma \ref{l:maximumwork} and Lemma \ref{lem:loadbigt} to obtain a simple schedulability test condition for task set $\tau$ on unit-speed processors.

%
%
%
%
%
%
%

\begin{lemma}\label{thm:tighter}
		A task set $\tau$ satisfying necessity conditions (\ref{eq:L<D}) and (\ref{eq:U<m}) is schedulable by G-RM on $m$ unit-speed processors if the following condition is satisfied:
	\begin{equation}
\sum_{\forall \tau_i \in \tau \atop u_{i} > 1} \frac{2u_{i} - \gamma_{i}}{2 - \gamma_{i}} + \sum_{\forall \tau_i \in \tau \atop u_{i}\leq 1} u_{i} \leq  m-\gamma_{max}(m-2)-U_{\sum} 
	\label{eq:G-RM}
	\end{equation}
\end{lemma}

\begin{proof}
		We prove the lemma by contradictions. Suppose that conditions in Lemma \ref{eq:G-RM} are satisfied and a task set $\taskset$ is failed by G-RM on $m$ unit-speed processors.

	Then from Lemma \ref{l: keylemma} and Lemma \ref{l:maximumwork}, we know there exists a task $\tau_k$ and a time interval of $I$, where $|I|\geq T_k$ and	

	\[
	\sum_{\forall \tau_i \in \tau \atop T_i \leq T_k}work(\tau_k,|I|,1)>\left( m-(m-1)\gamma_{max}-\frac{\displaystyle\sum_{\forall i, i \neq k \atop T_i\leq T_k} C_i}{|I|}\right)|I|
	\]
	
	Since $|I|\geq T_k$, $\frac{C_i}{|I|}\leq u_i$ when $T_i\leq T_k$. Therefore 
	\[
	\frac{\sum_{\forall i, i \neq k \wedge T_i\leq T_k} C_i}{|I|}\leq \sum_{\forall i, i \neq k \wedge T_i\leq T_k} u_i.
	\]
	Then 
	\[
	\sum_{\forall \tau_i \in \tau, T_i \leq T_k} \frac{work(\tau_k,|I|,1)}{|I|}>\left( m-(m-1)\gamma_k- \sum_{\forall i, i \neq k \atop T_i\leq T_k} u_i\right)
	\]
	
Combining with Lemma \ref{lem:loadbigt} we have
\[
	\sum_{T_i\leq T_k \atop u_{i} > 1} \frac{2u_{i} - \gamma_{i}}{2 - \gamma_{i}} + \sum_{T_i\leq T_k \atop u_{i}\leq 1} u_{i} > m-\gamma_{k}(m-1)-\sum_{\forall i, i \neq k \atop T_i\leq T_k} u_i
\]

	Then it must hold that
\begin{eqnarray*} \label{eq:midresult}
&&\sum_{\forall \tau_i \in \tau \atop u_{i} > 1} \frac{2u_{i} - \gamma_{i}}{2 - \gamma_{i}} + \sum_{\forall \tau_i \in \tau \atop u_{i}\leq 1} u_{i} > m-\gamma_{k}(m-1)-\sum_{\forall i, i \neq k \atop T_i\leq T_k} u_i \\
&>& m-\gamma_{k}(m-1)-(U_{\sum}-u_k) \\
&\geq& m-\gamma_{k}(m-1)-(U_{\sum}-\gamma_k) \\
&=& m-\gamma_{k}(m-2)-U_{\sum} \\
&\geq& m-\gamma_{max}(m-2)-U_{\sum} 
\end{eqnarray*}

%
\end{proof}

If all tasks are light and we treat them all as sequential tasks, i.e, $\forall \tau_k \in \tau : u_{k} = \gamma_{k} \leq 1$, then (\ref{eq:G-RM}) perfectly degrades to the classical utilization bound of G-RM for scheduling sequential tasks \cite{Bertogna2005New}:

\begin{equation}\label{e:sequential1}
U_{\sum} \leq \frac{(1-u_{max})m}{2}+u_{max}.
\end{equation}

Moreover, if all tasks are light but we treat them as parallel tasks, i.e.,$\forall \tau_k \in \tau : u_{k} = \gamma_{k} \leq 1$, then (\ref{eq:G-RM}) degrades to 

\begin{equation}\label{eq:sequential2}
U_{\sum} \leq \frac{(1-\gamma_{max})m}{2}+\gamma_{max}.
\end{equation}

Counter-intuitively, the schedulability test condition (\ref{eq:sequential2}) can accept more task sets
than (\ref{e:sequential1}) since $\gamma_{max}$
is in general smaller than $u_{\max}$. 

Therefore, 
the parallelism of light tasks is indeed useful to improve the schedulability under G-RM (the state-of-the-art techniques treat light tasks as sequential tasks). 

At last, without distinguishing heavy and light tasks, we derive a utilization-tensity bound for G-RM as follows.

\begin{theorem}\label{thm:ub}
	A task set $\tau$ is schedulable under G-RM on $m$ unit-speed processors if $\forall \tau_i \in \tau : L_i \leq T_i$ and 

		\begin{equation}\label{eq:G-RMtigher}
		U \leq \frac{(1-\gamma_{max})(2-\gamma_{max})}{4-\gamma_{max}}
		\end{equation}

	where $\displaystyle \gamma_{max} = \max_{\tau_i \in \tau} \{\gamma_{i}\}$ is the maximum tensity and $U=U_{\sum}/m$ is the normalized utilization of the task system $\tau$.
\end{theorem}

\begin{proof}
	We prove the theorem by contradictions. Suppose that conditions in Theorem \ref{thm:ub} are satisfied and a task set $\taskset$ is failed by G-RM on $m$ unit-speed processors.

	Then from Lemma \ref{thm:tighter} we know 
	\begin{equation}\label{eq:temp3}
\sum_{\forall \tau_i \in \tau \atop u_{i} > 1} \frac{2u_{i} - \gamma_{i}}{2 - \gamma_{i}} + \sum_{\forall \tau_i \in \tau \atop u_{i}\leq 1} u_{i} >  m-\gamma_{max}(m-2)-U_{\sum} 
	\end{equation}

	Since for any task $\tau_i \in \taskset$, $\gamma_{max} \geq \gamma_{i}$, We have
	\begin{eqnarray}
	&& \sum_{u_{i} > 1} \frac{2u_{i}-\gamma_{i}}{2-\gamma_{i}} +
	\sum_{u_{i} \leq 1} u_{i} \leq \sum_{u_{i} > 1} \frac{2u_{i}-\gamma_{i}}{2- \gamma_{max}} + \sum_{u_{i} \leq 1} u_{i} \nonumber \\
	&=&   \frac{ 2\sum_{u_{i}>1}u_i-\sum_{u_{i} > 1} \gamma_{i}+2\sum_{u_{i}\leq 1}u_i- \gamma_{max} \sum_{u_i \leq 1} u_i  }{2 - \gamma_{max}} \nonumber   \\
	&=&   \frac{2 U_{\sum}-( \sum_{u_{i} > 1} \gamma_{i} + \gamma_{max} \sum_{u_i \leq 1} u_i  )}{2 - \gamma_{max}} \nonumber   \\
	&\leq & \frac{ 2 U_{\sum}}{2 - \gamma_{max}} \nonumber
	\end{eqnarray}

	Then (\ref{eq:temp3}) implies:
	\begin{eqnarray*}
	&&\frac{ 2 U_{\sum}}{2 - \gamma_{max}} > m-\gamma_{max}(m-2)-U_{\sum} \nonumber \\
	&\Leftrightarrow& \frac{ (4-\gamma_{max}) U_{\sum}}{2 - \gamma_{max}} > m-\gamma_{max}(m-2)  \\
	&\Leftrightarrow& U_{\sum} > \frac{(1-\gamma_{max})(2-\gamma_{max})m+2\gamma_{max}(2-\gamma_{max})}{4-\gamma_{max}} \\
		&\Rightarrow& U_{\sum} > \frac{(1-\gamma_{max})(2-\gamma_{max})m}{4-\gamma_{max}}  \\
			&\Leftrightarrow& U > \frac{(1-\gamma_{max})(2-\gamma_{max})}{4-\gamma_{max}} 
	\end{eqnarray*}
	
Reaching a contradiction with the assumption that (\ref{eq:temp3}) holds, and the theorem is proved.	
\end{proof}

\subsection{Capacity Augmentation Bound}\label{sec:CAB}
The best known capacity augmentation bound for our considered problem is $2+\sqrt{3}\approx 3.732$ \cite{li2014analysis}.
In the following we improve it to $\frac{\sqrt{33}+7}{4}\approx 3.186$ based on the utilization-tensity bound in Theorem \ref{thm:ub}.

%
%
\begin{theorem}
	G-RM has a capacity augmentation bound of $\frac{\sqrt{33}+7}{4}$ for scheduling DAG tasks with implicit deadlines on $m$ unit-speed processors.	

\end{theorem}
\begin{proof}

 From definition \ref{l:capacitybound}, we need to prove that $\tau$ is schedulable if the following condition holds:
	
	\begin{equation}\label{eq: conditioncapacity}
	\forall \tau_i \in \tau : L_{i} \leq T_i/(\frac{\sqrt{33}+7}{4})  ~\wedge~ U_{\sum}  \leq  m / (\frac{\sqrt{33}+7}{4})
   \end{equation}
   
   From condition (\ref{eq: conditioncapacity}), we have $\gamma_{max}\leq \frac{7-\sqrt{33}}{4}$ and $U\leq \frac{7-\sqrt{33}}{4}$. 
   
   Since the value of $\frac{(1-\gamma_{max})(2-\gamma_{max})}{4-\gamma_{max}}$ is monotonically decreasing as $\gamma_{max}$ when $\gamma_{max}\leq 4-\sqrt{6} \approx 1.6$, we have 
   
   \begin{eqnarray*}
   \frac{(1-\gamma_{max})(2-\gamma_{max})}{4-\gamma_{max}} &\geq& \frac{(1-\frac{7-\sqrt{33}}{4})(2-\frac{7-\sqrt{33}}{4})}{4-\frac{7-\sqrt{33}}{4}} \\
   &=&\frac{7-\sqrt{33}}{4}
   \end{eqnarray*}
when $\gamma_{max}\leq \frac{7-\sqrt{33}}{4}$.
Since $U\leq \frac{7-\sqrt{33}}{4}$, we have
\[
	U \leq \frac{7-\sqrt{33}}{4} \leq \frac{(1-\gamma_{max})(2-\gamma_{max})}{4-\gamma_{max}}
\]

By Corollary \ref{thm:ub}, we know $\tau$ is schedulable.
\end{proof}

The graphical representation of the allowed parameter space $U \leq \frac{(1-\gamma_{max})(2-\gamma_{max})}{4-\gamma_{max}}$ is the area below the orange curve in Fig. \ref{f:bounds}, which is substantially larger than the allowed parameter space of the capacity augmentation bound (the area in the orange dash rectangular). we can see that task sets with small tensity but high utilization and task sets with high tensity but small utilization can still be accepted. In general, the smaller $\gamma_{max}$ is, the higher total utilization can be tolerated (also the other way around).

It is also interesting to get a better understanding of the
absolute and relative performance of G-EDF and G-RM for
DAG task systems. Besides for the empirical evaluations, in
previous work, capacity augmentation bound and speed up
factor are used as the theoretic metrics to evaluate the relative
performance of G-EDF and G-RM for parallel tasks, and in
general G-EDF is considered to perform better than G-RM
since G-EDF has both capacity augmentation bound and speed
up factor smaller than G-RM. However, this impression may
be misleading.

As shown in Fig. \ref{f:bounds}, G-EDF can tolerant a task set with greater normalized utilization than G-RM with small tensity whereas G-RM can tolerant a task set with greater normalized utilization than G-EDF when task set have great tensity.

\section{EXPERIMENTS}
\label{sec:evaluation}

\begin{figure*}
	\centering
	\subfigure[normalized utilization]{\includegraphics[width=2.3in]{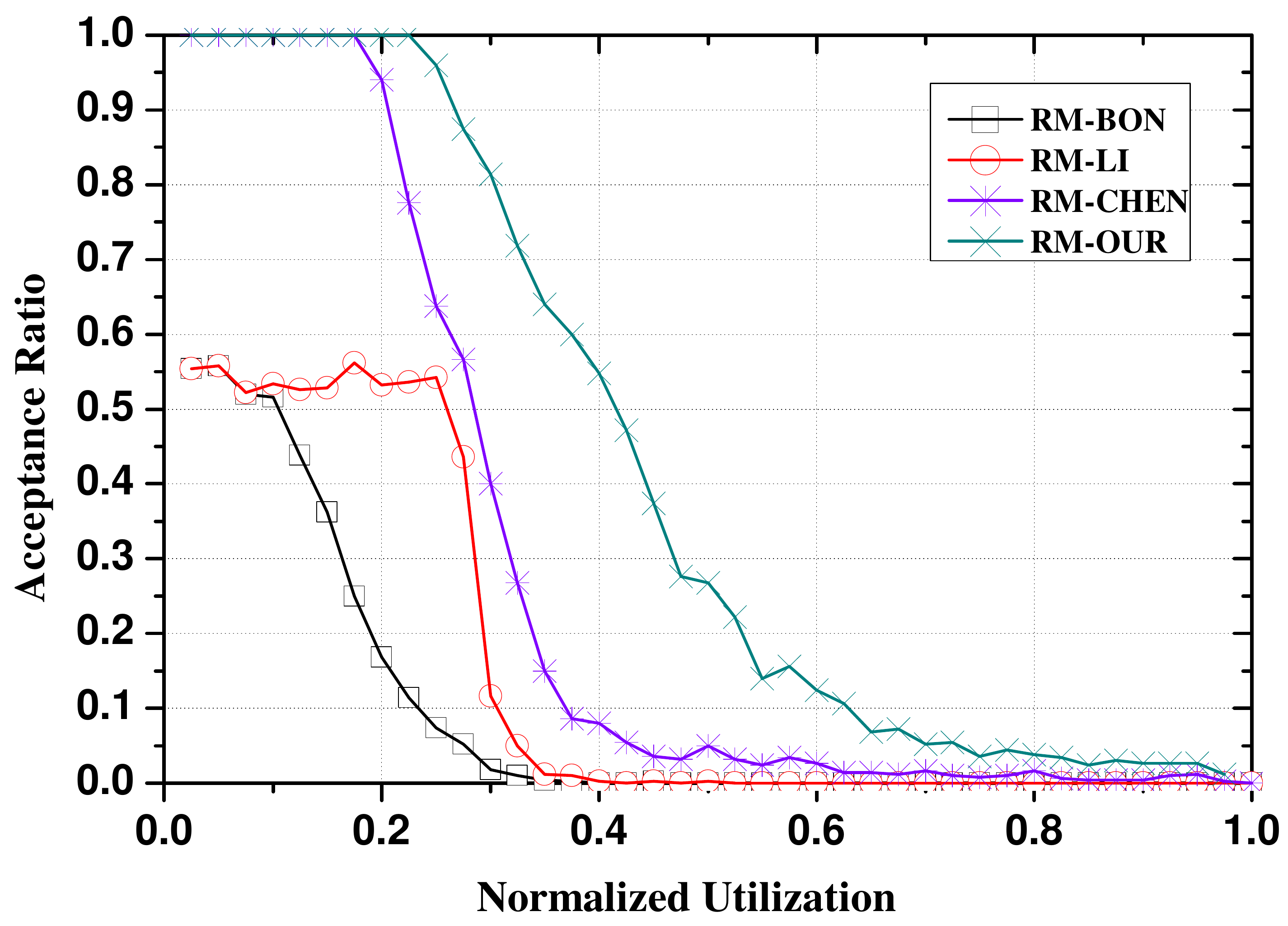}}
	\hspace{0.01in}
	\subfigure[upper bound of tensity]{\includegraphics[width=2.3in]{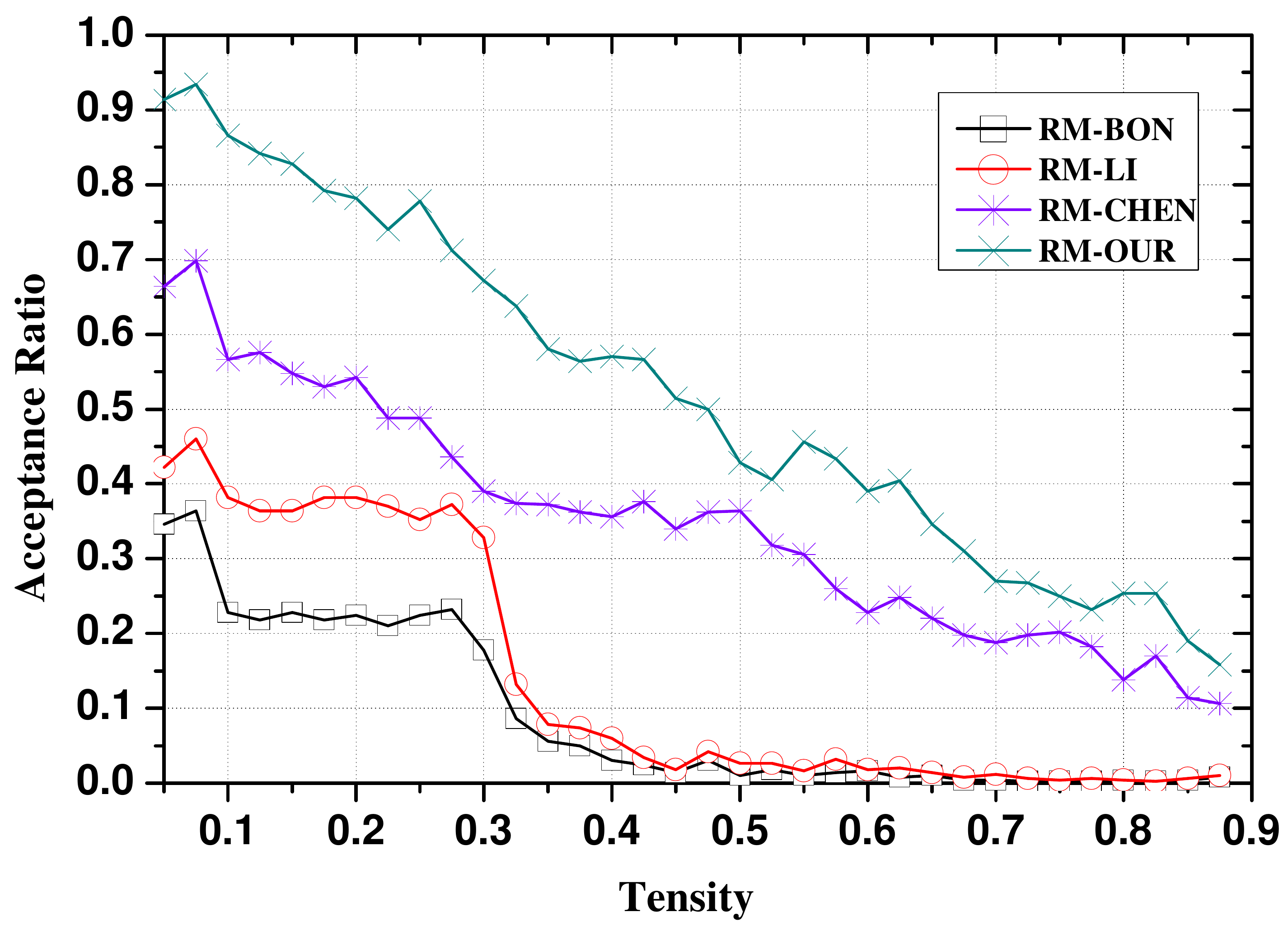}}
	\hspace{0.01in}
	\subfigure[number of tasks]{\includegraphics[width=2.3in]{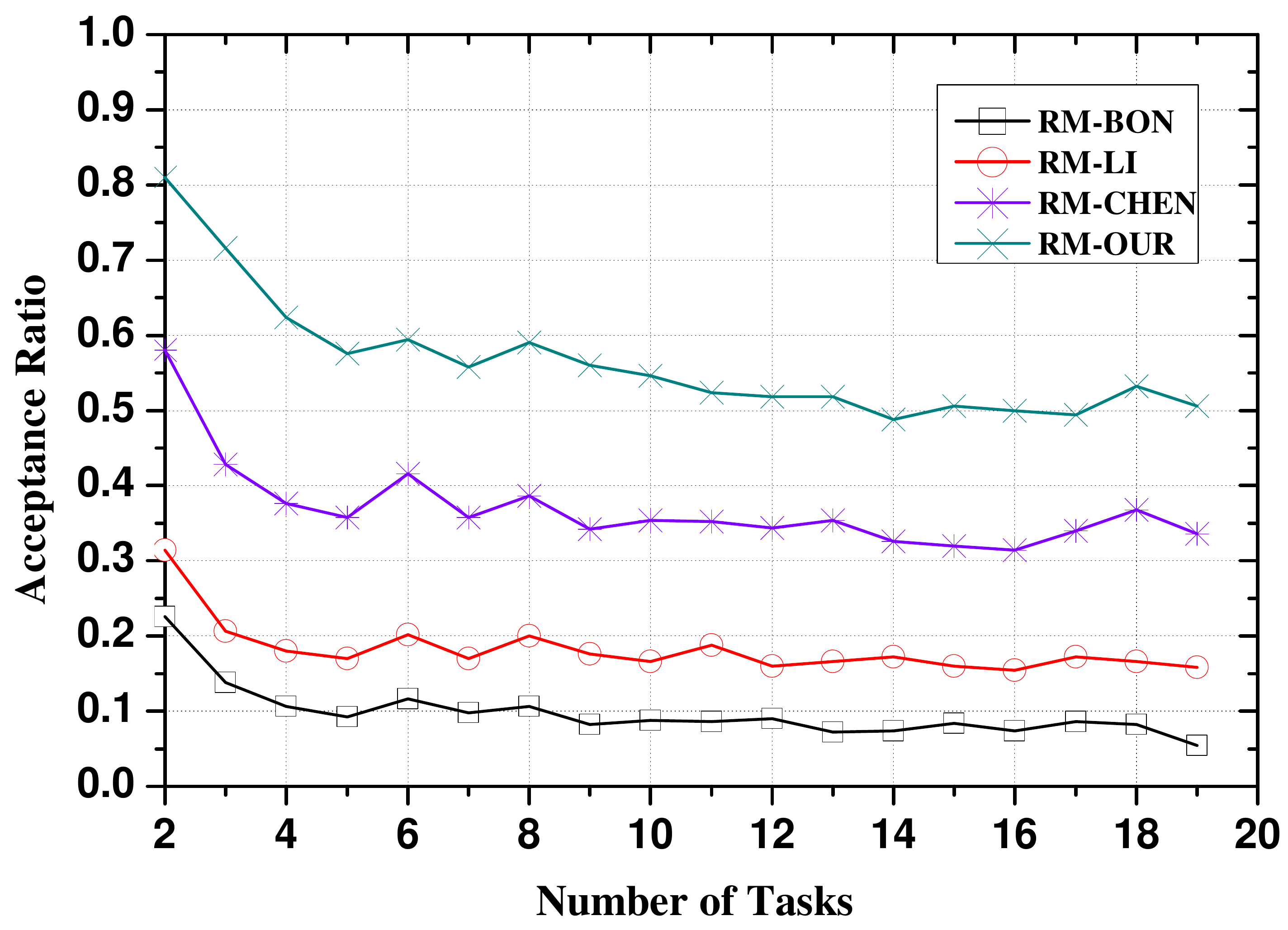}}
	\caption{Comparison of acceptance ratio of G-RM scheduling in different dimensions.}
	\label{fig:heavy}
\end{figure*}

In this section, we evaluate the performance of our proposed methods. In particular, we compare our tests in Theorem \ref{thm:ub}, denoted as RM-OUR, with other state-of-the-art G-RM schedulability bounds: (i) The schedulability test based on capacity augmentation bounds for G-RM scheduling in \cite{li2014analysis}, denoted by RM-LI. \revise{(ii) The utilization-based test in \cite{chen2014capacity}, denoted by RM-CHEN.} (iii) The schedulability test for G-DM in \cite{bonifaci2013feasibility}, which can also be used for G-RM by setting the deadline of each task equal to its period, denoted by RM-BON. 

Other methods for G-RM not included in our comparison either have high complexity (with respect to linear-time) or require intra-structure information. Recall that, our focus in this paper is to provide schedualbility bounds.

The task sets are generated using the Erd\"{o}s-R\'{e}nyi method $G(n_i , p)$ \cite{cordeiro2010random}.
For each task, the number of vertices is randomly chosen in the range $[50, 150]$.
The worst-case execution time of each vertex
is randomly picked in the range $[20, 50]$. For each possible edge we generate a random value in the range $[0, 1]$
and add the edge to the graph only if the generated value is less than a predefined threshold $p$. The same as in \cite{saifullah2014parallel},
we also add a minimum number of additional edges to make a task graph weakly connected. \revise{The period of each task $\tau_i$ is generated according to its tensity $L_i/T_i$. For each task set, we first generate an upper bound of tensity $\gamma_{up}$ which is picked in the range of $(0, 1)$. Then for each task $\tau_i$, its tensity is randomly chosen in the range of $(0, \gamma_{up}]$.} For each task set, we randomly generate $n$ tasks, where $n$ is in the range $[2, 10]$. For each parameter configuration, we generate 1000 task sets.


Fig.\ref{fig:heavy}.(a) compares the acceptance ratio of task sets with different normalized utilization,  where $\gamma_{up}$ is randomly picked from $[0.1, 0.6]$ and $n$ is randomly chosen from $[2, 10]$. The number of processors is calculated by $m=\lceil\sum\frac{C_i}{T_i}/ U \rceil$ where $U$ is the normalized utilization indicates the $x$ axis in Fig.\ref{fig:heavy}.(a). 
Fig.\ref{fig:heavy}.(b) compares the acceptance ratio of task sets with different $\gamma_{up}$. Fig.\ref{fig:heavy}.(b) follows the same setting as Fig.\ref{fig:heavy}.(a), but task periods are generated according to different value of $\gamma_{up}$ (corresponding to the x-axis). The normalized utilization of each task set is randomly chosen from $[0.1, 0.6]$. Fig.\ref{fig:heavy}.(c) compares the acceptance ratio of task sets with different number of tasks. Fig.\ref{fig:heavy}.(c) follows the same setting as Fig.\ref{fig:heavy}.(a), but tasks are generated with different number of tasks $n$ (corresponding to the x-axis) and the normalized utilization of each task set is randomly chosen from $[0.1,6]$. Since the acceptance ratio could be always high for some tests at a low normalized utilization/tensity and could be quite low for some tests at a high normalized utilization/tensity, e.g., no task sets are schedulable under RM-LI if $U>0.5$, choosing a fixed utilization/tensity will be hard to identify the trends. Thus we choose normalized utilization/tensity in a range to better illustrate the difference between different approaches.

In Fig.\ref{fig:heavy}.(a), we can see the schedulability of all tests decrease as the normalized utilization increases. In Fig.\ref{fig:heavy}.(b), we can see the schedulability of all tests decrease as the tensity increases. These two results are consistent with that observed from our utilization-tensity bound. In Fig.\ref{fig:heavy}.(c), we can see the schedulability of all tests decrease as the number of tasks increases. The reason is that, when other parameters are fixed, the more tasks the greater maximum tensity could be, and task with great tensity hurts the schedulability. In general, the experiment results show that our utilization-tensity test outperforms other tests in different dimensions. 


\section{Related Work}  
\label{sec:relate}
In this section, we introduce the state-of-the-art results about the theoretical bounds for real-time DAG task systems, as well as techniques analyzing G-RM.

\emph{Resource augmentation bound} (also called speedup factor) \cite{bonifaci2013feasibility} and \emph{capacity augmentation bound} \cite{li2014analysis} are two widely used metrics to evaluate the quality of scheduling algorithms and analysis methods for DAG task systems. Resource augmentation bound is a \emph{comparative} metric with respect to the optimal schedulers and cannot be directly used as a schedulability test, whereas capacity augmentation bound can be directly used for schedulability test.

For multiple DAGs with implicit deadlines, Li et.al \cite{li2014analysis} proved a capacity augmentation bound of $2$ under federated scheduling, and proved a resource augmentation bound of $2$ with respect to hypothetical optimal scheduling algorithms with implicit-deadline DAG tasks. For mixed-criticality DAGs with implicit deadlines, Li et.al \cite{li2017mixed} proved that for high utilization tasks, the mixed criticality federated scheduling has a capacity augmentation bound of $2+2\sqrt{2}$ and $\frac{5+\sqrt{5}}{2}$ for dual- and multi-criticality systems, respectively. Moreover, they also derived a capacity augmentation bound of $\frac{11m}{3m-3}$ for dual-criticality systems with both high- and low-utilization tasks. Chen \cite{chen2016federated} showed that any federated scheduling algorithm has a resource augmentation bound of at least $\Omega(\min\{m,n\})$ with respect to any optimal scheduling algorithm, where $n$ is the number of tasks and $m$ is the number of processors. Baruah proved a resource augmentation bound of $3-\frac{1}{m}$ for constrained deadline DAG tasks \cite{baruah2015federated} and proved a resource augmentation bound of $4-\frac{2}{m}$ for arbitrary deadline DAG tasks \cite{baruah2015federated1}.

For global scheduling,  Baruah et.al \cite{baruah2012generalized} proved a resource augmentation bound of $2$ under Global Earliest Deadlines First (G-EDF)
for a single recurrent DAG task with an arbitrary deadline. For multiple DAG tasks with arbitrary deadlines, Li et.al \cite{li2013outstanding} and Bonifaci et.al \cite{bonifaci2013feasibility} proved a resource augmentation bound of $2-\frac{1}{m}$ under G-EDF, and Bonifaci et.al \cite{bonifaci2013feasibility} proved a resource augmentation bound of $3-\frac{1}{m}$ under Global Deadline Monotonic scheduling (G-DM). For multiple DAGs with implicit deadlines, Li et.al \cite{li2013outstanding} proved a bound of $4-\frac{2}{m}$ under GEDF, and this bound is further improved to $\frac{3+\sqrt{5}}{2}$, which is proved to be tight when the number $m$ of processors is sufficiently large. 	
Moreover, Li et.al \cite{li2014analysis} proved a bound of $2+\sqrt{3}$ under Global Rate Monotonic (G-RM). 

The capacity augmentation bounds for \emph{decomposition-based} global scheduling are
restricted to implicit-deadline DAG tasks. Earlier work began with synchronous tasks (a special case of DAG tasks). For a restricted set of synchronous tasks, Lakshmanan et.al \cite{lakshmanan2010scheduling} proved a bound of $3.42$ using G-DM for decomposed tasks. For more general synchronous tasks, Saifullah et.al \cite{saifullah2013multi} proved a capacity augmentation bound of $4$ for G-EDF and $5$ for G-RM. For DAG tasks, Saifullah et.al \cite{saifullah2014parallel} proved a capacity augmentation bound of $4$ under G-EDF on decomposed tasks, and Jiang et.al \cite{jiang2016decomposition} refined this bound to the range of $[2-\frac{1}{m},4-\frac{2}{m})$, depending on the DAG structure characteristics. 

Besides theoretical bounds, there are also other schedulability tests proposed which require different pre-knowledge and have different complexity, to analyze DAG tasks under fixed priority scheduling. In \cite{melani2015response}, response analysis techniques for DAG tasks with constrained deadlines under both G-EDF and fixed priority scheduling were presented. Parri et.al \cite{Parri2015Response} analyzed the scheduling of DAG task with arbitrary deadlines under both G-EDF and G-DM. Bonifaci et.al \cite{bonifaci2013feasibility}. Chen et.al \cite{chen2014capacity} provided polynomial-time schedulability tests for DAG tasks with arbitrary deadlines under both G-EDF and G-RM.

\section{Conclusions}

%
	In this paper, we develop new analysis techniques for determining the schedulability of DAG tasks with implicit deadlines under G-RM. Especially, we derive a utilization-based schedulability test with linear-time under G-RM, which can perfectly degrade to the well-know utilization-based bound for scheduling sequential tasks under G-RM\cite{Bertogna2005New}, when all tasks have utilization no greater than 1 and treated as sequential tasks. We also present an utilization-tensity bound for scheduling DAG tasks under G-RM, which is the most tight one in the-state-of-the-art, and we prove a capacity augmentation bound of $3.186$ for G-RM, which is tighter than the best known result in the-state-of-the-art, i.e., $2+\sqrt{3}\approx 3.732$. The experiment results indicate that our new algorithm outperforms other schedulability tests for G-RM.   
	
	There are several directions of future work. Clearly, the current bounds derived here may not be tight. Therefore, the minimum of the capacity augmentation required by G-RM is still open. Moreover, techniques proposed in this paper are developed for tasks with implicit deadlines, we would like to generalize them to arbitrary deadline tasks.

\bibliographystyle{IEEEtran}
\bibliography{IEEEabrv,reference}

\end{document}